\documentclass{LMCS}

\def\doi{8(3:25)2012}
\lmcsheading%
{\doi}
{1--21}
{}
{}
{Oct.~14, 2011}
{Sep.~29, 2012}
{}
 
\pdfoutput=1
\usepackage{amsmath}
\usepackage{amssymb}
\usepackage{amsfonts}
\usepackage{graphicx}
\usepackage{color}
\usepackage{algorithm2e}
\usepackage{enumerate,hyperref}
\usepackage{multirow}
\usepackage{wrapfig}
\usepackage[small,bf]{caption}

\newcommand{\ox}{X}

\newcommand{\obp}{B_P}
\newcommand{\defn}{\overset{\triangle}{=}}
\newcommand{\FF}{\mathit{F}}
\newcommand{\TT}{\mathit{T}}
\newcommand{\precond}{\mathit{Pre}}
\newcommand{\MEM}{\mathit{MEM}}
\newcommand{\EQ}{\mathit{EQ}}
\newcommand{\YES}{\mathit{YES}}
\newcommand{\NO}{\mathit{NO}}
\newcommand{\timplies}{\Rightarrow}
\newcommand{\nimplies}{\nRightarrow}
\newcommand{\valueof}[1]{[\![{#1}]\!]}
\newcommand{\symbols}[1]{\sigma({#1})}
\newcommand{\goesto}[1]{\overset{{#1}}{\longrightarrow}}

\newcommand{\Prop}{\mathit{QF}}
\newcommand{\Bool}{\mathit{Bool}}
\newcommand{\CEX}{\mathit{CEX}}
\newcommand{\set}[1]{\{#1\}}
\newcommand{\Val}[1]{\mathit{Val}_{{#1}}}
\newcommand{\assertion}[1]{\mathtt{\{}\ #1\ \mathtt{\}}}
\newcommand{\tab}{\ \ \ \ }
\newcommand{\restrict}[1]{\!\!\downarrow_{{#1}}}
\newcommand{\false}{\texttt{F}}
\newcommand{\true}{\texttt{T}}

\newcommand{\zap}[1]{#1}
\newcommand{\etal}{\textit{et al.}}

\definecolor{light-gray}{gray}{0.8}

\newcommand{\shade}[1]{\colorbox{light-gray}{\color{black}{#1}}}

\begin{document}

\title[Predicate Generation for Learning-Based Invariant
Inference]{Predicate Generation for Learning-Based Quantifier-Free
  Loop Invariant Inference\rsuper*}

\author[W.~Lee]{Wonchan Lee\rsuper a}
\address{{\lsuper{a,d}}Seoul National University, Korea}
\email{\{wclee, kwang\}@ropas.snu.ac.kr}

\author[Y.~Jung]{Yungbum Jung\rsuper b}
\address{{\lsuper b}Fasoo.com, Korea and Seoul National University, Korea}
\email{yb@fasoo.com}

\author[B.~Wang]{Bow-Yaw Wang\rsuper c}
\address{{\lsuper c}Academia Sinica, Taiwan}
\email{bywang@iis.sinica.edu.tw}

\author[K.~Yi]{Kwangkuen Yi\rsuper d}
 

\titlecomment{{\lsuper *}This paper is a revised and extended version of the
  paper ``Predicate Generation for Learning-Based Quantifier-Free Loop
  Invariant Inference'' that has been published in the proceedings of
  TACAS 2011~\cite{JLWY:11:PGLBQFLII}. This work was supported by the
  Engineering Research Center of Excellence Program of Korea Ministry
  of Education, Science and Technology(MEST) / National Research
  Foundation of Korea(NRF) (Grant 2012-0000468), National Science
  Council of Taiwan Grant Numbers 99-2218-E-001-002-MY3 and
  100-2221-E-002-116-, National Science Foundation (award
  no. CNS0926181), and by Republic of Korea Dual Use Program
  Cooperation Center(DUPC) of Agency for Defense Development(ADD)}

\maketitle

\begin{abstract}
  We address the predicate generation problem in the context of loop
  invariant inference. Motivated by the interpolation-based
  abstraction refinement technique, we apply the interpolation theorem
  to synthesize predicates implicitly implied by program texts. Our
  technique is able to improve the effectiveness and efficiency of the
  learning-based loop invariant inference algorithm of Jung, Kong,
  Wang and Yi (2010). We report experimental results of examples from
  Linux, SPEC2000, and the Tar utility.
\end{abstract}

\section{Introduction}
\label{section:introduction}


One way to prove that an annotated loop satisfies its pre-
and post-conditions is by giving loop invariants. 
In an annotated loop, pre- and post-conditions specify intended
effects of the loop. The actual behavior of the annotated loop
however does not necessarily conform to its specification. Through
loop invariants, verification tools can check whether the
annotated loop fulfills its specification
automatically~\cite{FM:04:MPVC}.


Finding loop invariants is tedious and sometimes requires
intelligence. Recently, an automated technique based on algorithmic
learning and predicate abstraction is proposed~\cite{VMCAI10}. Given a
fixed set of atomic predicates and an annotated loop, the
learning-based technique can infer a quantifier-free loop invariant
over the given atomic predicates. By employing a learning algorithm
and a mechanical teacher, the new technique is able to generate loop
invariants without constructing abstract models nor computing fixed
points.


As in other techniques based on predicate abstraction, the selection
of atomic predicates is crucial to the effectiveness of the
learning-based technique. Oftentimes, users extract atomic predicates
from program texts heuristically. If this simple strategy does not
yield necessary atomic predicates to express any loop invariants the
loop invariant inference algorithm will not be able to infer a loop
invariant. Even when the heuristic does give necessary atomic
predicates, it may select too many redundant predicates and impede the
efficiency of loop invariant inference algorithm.


One way to circumvent this problem is to generate atomic predicates by
need.  Several techniques have been developed to synthesize atomic
predicates by
interpolation~\cite{EsparzaKS06,Jhala06,McMillan:05:ITP,McMillan06}. Let
$A$ and $B$ be logic formulae. An interpolant $I$ of $A$ and $B$ is a
formula such that $A \Rightarrow I$ and $I \wedge B$ is
inconsistent. Moreover, the non-logical symbols in $I$ must occur in
both $A$ and $B$.  By Craig's interpolation theorem, an interpolant
$I$ always exists for any first-order formulae $A$ and $B$ when $A
\wedge B$ is inconsistent~\cite{craig}. The interpolant $I$ can be
seen as a concise summary of $A$ with respect to $B$. Indeed, many
abstraction refinement techniques for software model
checking~\cite{EsparzaKS06,POPL04,Jhala06,McMillan:05:ITP,McMillan06}
have used interpolation to synthesize atomic predicates.


Inspired by the refinement technique in software model checking, we
develop an \\ interpolation-based technique to synthesize atomic
predicates in the context of learning-based loop invariant
inference. Our algorithm does not add new atomic predicates by
interpolating invalid execution paths in control flow graphs.  We
instead interpolate the loop body with purported loop invariants from
the learning algorithm. We adopt the existing interpolating theorem
provers~\cite{csisat,princess,mathsat4,McMillan:05:ITP} for the
interpolation. With our new predicate generation technique, we can
improve the effectiveness and efficiency of the existing
learning-based loop invariant inference
technique~\cite{VMCAI10}. Constructing the set of atomic predicates is
fully automatic and on-demand.


\noindent

\subsection{Example}
Consider the following annotated loop:

\begin{align*}
  &\mathtt{\{}\ n \geq 0 \wedge x = n \wedge y = n\ \mathtt{\}}\  \\
  &\mathtt{while}\ x > 0\ \mathtt{do}\ \\
  &\quad x = x - 1;\ y = y - 1\ \\
  &\mathtt{done}\ \\
  &\mathtt{\{}\ x + y = 0\ \mathtt{\}}
\end{align*}
\\
Assume that variables $x$ and $y$ both have the value $n \geq 0$
before entering the loop. The loop body decreases each variable by one
until the variable $x$ becomes zero. We want to show that $x + y$ is
zero after executing the loop. This requires of us to establish the
fact that variables $x$ and $y$ have the same value during iterations
and eventually become zero after exiting the loop. To express this
fact as a loop invariant, we require a predicate $x = y$. The program
text however does not reveal this equality explicitly. Moreover,
atomic predicates from the program text cannot express any loop
invariant that establishes the given specification. Using atomic
predicates in the program text is not sufficient in this
case. However, we can exploit the fact that any loop invariant $\iota$
should be weaker than the pre-condition $\delta$ and stronger than the
disjunction of the loop guard $\kappa$ and the post-condition
$\epsilon$ ($\delta \timplies \iota \timplies \kappa \vee
\epsilon$). Then, we can gen an interpolant from inconsistent formula
$\delta \wedge \neg (\kappa \vee \epsilon)$ and extract atomic
predicates in it. From the interpolant of $(n \geq 0 \wedge x = n
\wedge y = n) \wedge \neg (x > 0 \vee x + y = 0)$, we obtain two
atomic predicates $x = y$ and $2y \geq 0$. Observe that the
interpolation is able to synthesize the necessary predicate $x =
y$. In fact, loop invariant $x = y \wedge x \ge 0$ establishes the
specification of the loop.

\noindent

\subsection{Related Work}
Jung \etal~\cite{VMCAI10} introduce the loop invariant inference
technique based on algorithmic learning. Kong \etal~\cite{APLAS10}
extend this technique to quantified loop invariant inference. Both
algorithms require users to provide atomic predicates. The present
work addresses this problem for the case of quantifier-free loop
invariants.

Recently, Lee \etal~\cite{LWY:12:TAAL} introduce learning-based
technique for termination analysis. The technique infers the
transition invariant of a given loop as a proof of termination, by
combining algorithmic learning and decision procedures. In the paper,
the authors design a heuristic to generate atomic transition
predicates. It is an interesting future work to adapt our technique in
the present paper for transition invariant inference.

Many interpolation algorithms and their implementations are
available~\cite{csisat,princess,mathsat4,McMillan:05:ITP}.
Interpolation-based techniques for predicate refinement in software
model checking are proposed
in~\cite{EsparzaKS06,POPL04,Jhala06,Jhala07,McMillan06}. Abstract
models used in these techniques however may require excessive
invocations to theorem provers. Another interpolation-based technique
for first-order invariants is developed in~\cite{McMillan}. The
paramodulation-based technique presented in the paper does not
construct abstract models as our approach. It however only generates
invariants in first-order logic with equality.
A template-based predicate generation technique for quantified
invariants is proposed~\cite{PLDI09}. The technique reduces the
invariant inference problem to constraint programming and generates
predicates in user-provided templates.




\subsection{Paper Organization}
Section~\ref{section:preliminaries} gives preliminaries for the
presentation. Section~\ref{section:inferring-unquantified-loop-invariants}
reviews the learning-based loop invariant inference
framework~\cite{VMCAI10}. Section~\ref{section:predicate-generation-by-interpolation}
presents our interpolation-based predicate generation technique.
Section~\ref{section:algorithm} presents the loop invariant inference
algorithms with automatic predicate
generation. Section~\ref{section:experimental-results} presents and
discusses our experimental results.  Section~\ref{section:conclusions}
concludes this work.

\section{Preliminaries}
\label{section:preliminaries}

\subsection{Quantifier-free Formulae}
\label{subsec:qf}

Let $\Prop$ denote the quantifier-free logic with equality, linear
inequality, and uninterpreted
functions. Define the 
\emph{domain} $\mathbb{D} = \mathbb{Q} \cup \mathbb{B}$ where
$\mathbb{Q}$ is the set of rational numbers and $\mathbb{B} = \{
\FF, \TT \}$ is the Boolean domain. Fix a set $\ox$ of variables.
A \emph{valuation} over $\ox$ is a function from
$\ox$ to  
$\mathbb{D}$. The class of valuations over $\ox$ is denoted by $\Val{\ox}$.
For any formula $\theta \in \Prop$ and valuation $\nu$ over free
variables in $\theta$, $\theta$ is
\emph{satisfied} by $\nu$ (written $\nu \models \theta$) if $\theta$
evaluates to $\TT$ under $\nu$; $\theta$ is \emph{inconsistent} if
$\theta$ is not satisfied by any valuation.
Given a formula $\theta \in \Prop$, a \emph{satisfiability modulo
theories (SMT) solver} returns a satisfying
valuation $\nu$ of $\theta$ if $\theta$ is not
inconsistent~\cite{mathsat4,Yices}.

\subsection{Interpolation Theorem}
\label{subsec:interpolation}
For $\theta \in \Prop$, we denote the set of non-logical symbols
occurred in $\theta$ by $\symbols{\theta}$. Let $\Theta = [\theta_1, 
\ldots, \theta_m]$ be a sequence with $\theta_i \in \Prop$ for $1 \leq i
\leq m$. The sequence $\Theta$ is \emph{inconsistent} if 
$\theta_1 \wedge \theta_2 \wedge \cdots \wedge \theta_m$ is inconsistent.
The sequence $\Lambda = [\lambda_0, \lambda_1, \ldots,
\lambda_m]$ of quantifier-free formulae is an \emph{inductive
  interpolant} of $\Theta$ if 
\begin{iteMize}{$\bullet$}
\item $\lambda_0 = \TT$ and $\lambda_m = \FF$;
\item for all $1 \leq i \leq m$, $\lambda_{i-1} \wedge \theta_i
  \Rightarrow \lambda_i$; and
\item for all $1 \leq i < m$, $\symbols{\lambda_i} \subseteq
  \symbols{\theta_i} \cap \symbols{\theta_{i+1}}$.
\end{iteMize}
The third condition of interpolants makes them attractive to use for
predicate generation; since the set of symbols in an interpolant
should be an intersection of sets of symbols in two inconsistent
formulae, it sometimes consists of predicates which do not appear in
the two. The interpolation theorem states that an inductive
interpolant exists for any inconsistent
sequence~\cite{craig,McMillan:05:ITP,McMillan06}. Some of existing
theorem provers~\cite{csisat,princess,mathsat4,McMillan:05:ITP}
can generate interpolants from inconsistent sequences.

\subsection{Predicate Abstraction}
\label{subsec:pred-abs}

\begin{figure*}[b]
  \centering
  \resizebox{0.4\textwidth}{!}{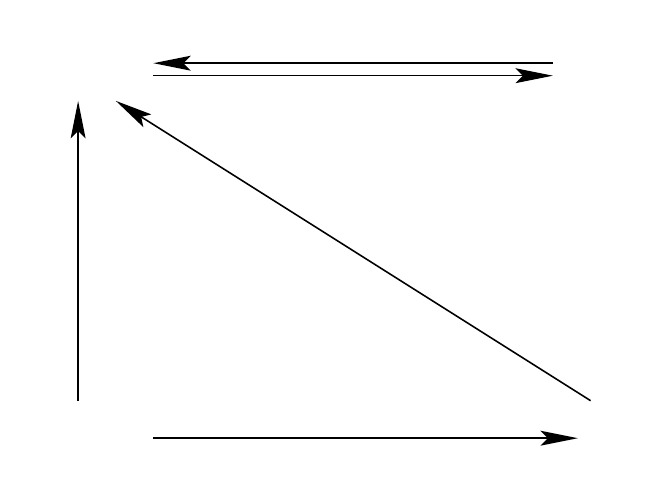}
  \caption{Relating $\Prop$ and $\Bool[\obp]$}
  \label{figure:domains}
\end{figure*}

Let $\Prop[P]$ denote the set of quantifier-free formulae over the set
$P$ of atomic predicates. A cube over $P$ is a conjunction $p_1 \land
\cdots \land p_k \land \neg p_{k+1} \land \cdots \land \neg p_{k+k'}$
where all $p_j \in P$ are distinct. We say that $k + k'$ is the size
of the cube. A minterm over $P$ is a cube whose size is $|P|$.

Consider the set $\Bool[\obp]$ of Boolean formulae over the set $\obp$
of Boolean variables where $\obp \defn \{ b_p : p \in P \}$. An
\emph{abstract valuation} is a function from $\obp$ to
$\mathbb{B}$. We write $\Val{\obp}$ for the set of abstract
valuations. A Boolean formula in $\Bool[\obp]$ is a \emph{canonical
  monomial} if it is a conjunction of literals, where each Boolean
variable in $\obp$ occurs exactly once. The following
functions~\cite{VMCAI10,JLWY:11:PGLBQFLII} relate formulae in
$\Prop[P]$ and $\Bool[\obp]$ (Figure~\ref{figure:domains}):

\begin{equation*}
  \begin{array}{rcl}
    \gamma (\beta) & \defn & \beta[\obp \mapsto P]\\
    \alpha (\theta) & \defn & 
    \bigvee \{ \beta \in \Bool[\obp] :
    \beta \mbox{ is a canonical monomial and }
    \theta \wedge \gamma (\beta) \mbox{ is satisfiable} \}\\
    \gamma^* (\mu) & \defn &
    \bigwedge\limits_{\mu (b_p) = \TT}
    \{ p \} \wedge
    \bigwedge\limits_{\mu (b_p) = \FF}
    \{ \neg p \}\\
    \alpha^* (\nu) & \defn & \mu \mbox{ where }
    \mu (b_p) =
    \left\{
      \begin{array}{ll}
        \TT & \mbox{ if } \nu \models p\\
        \FF & \mbox{ if } \nu \not\models p
      \end{array}
    \right. \\
    \Gamma (\nu) & \defn & 
    \bigwedge\limits_{x \in \ox} x = \nu (x)
  \end{array}
\end{equation*}
The abstraction function $\alpha$ maps any quantifier-free formula to
a Boolean formula in $\Bool[\obp]$, whereas the concretization
function $\gamma$ maps any Boolean formula in $\Bool[\obp]$ to a
quantifier-free formula in $\Prop[P]$. Moreover, the function
$\alpha^*$ maps a valuation over $\ox$ to a valuation over $\obp$; the
function $\gamma^*$ maps a valuation over $\obp$ to a quantifier-free
formula in $\Prop[P]$. The function $\Gamma (\nu)$ specifies the
valuation $\nu$ in $\Prop$. Observe that quantifier-free formula
$\gamma(\beta)$ is a minterm when Boolean formula $\beta$ is a
canonical monomial. Observe also that formula $\gamma(\alpha(\theta))$
is in disjunctive normal form and equivalent to $\theta \in \Prop[P]$.

Consider, for instance, $P = \{ n \geq 0, x = n, y =
n \}$ and $\obp = \{ b_{n \geq 0}, b_{x = n}, b_{y = n} \}$. We have
$\gamma (b_{n \geq 0} \wedge \neg b_{x = n}) = n \geq 0 \wedge \neg (x
= n)$ and 
\begin{equation*}
\alpha (\neg (x = y)) = 
\begin{array}{l}
(b_{n \geq 0} \wedge b_{x = n} \wedge \neg b_{y = n}) \vee 
(b_{n \geq 0} \wedge \neg b_{x = n} \wedge b_{y = n}) \vee\\
(b_{n \geq 0} \wedge \neg b_{x = n} \wedge \neg b_{y = n}) \vee
(\neg b_{n \geq 0} \wedge b_{x = n} \wedge \neg b_{y = n}) \vee\\
(\neg b_{n \geq 0} \wedge \neg b_{x = n} \wedge b_{y = n}) \vee
(\neg b_{n \geq 0} \wedge \neg b_{x = n} \wedge \neg b_{y = n}).
\end{array}
\end{equation*}
Moreover, $\alpha^* (\nu)(b_{n \geq 0}) = \alpha^* (\nu) (b_{x = n}) =
\alpha^* (\nu) (b_{y = n}) = \TT$ when $\nu (n) = \nu (x) = \nu (y) =
1$. And $\gamma^* (\mu) = n \geq 0 \wedge x = n \wedge \neg (y = n)$
when $\mu (b_{n \geq 0}) = \mu (b_{x = n}) = \TT$ but $\mu (b_{y = n})
= \FF$.

The following lemmas prove useful properties of these abstraction and
concretization functions. 

\begin{lem}
  \label{lemma:exact}
  Let $P$ be a set of atomic predicates,
  $\theta \in \Prop[P]$, and $\beta$ a canonical monomial in
  $\Bool[B_P]$. Then $\theta \wedge \gamma (\beta)$ is
  satisfiable if and only if $\gamma (\beta) \timplies
  \theta$.
\end{lem}

\begin{proof}
  Let $\theta' = \bigvee\limits_i \theta_i \in \Prop[P]$ be a formula
  in disjunctive normal form such that $\theta' \Leftrightarrow
  \theta$. Note that each $\theta_i$ is a cube over set $P$. Let
  $\mathit{Lit}(\theta)$ be a set of literals in formula
  $\theta$. Then, $\mathit{Lit}(\theta_i) \subseteq P \cup \set{\neg p
    : p \in P}$.

  Assume $\theta \wedge \gamma (\beta)$ is satisfiable. Then $\theta'
  \wedge \gamma (\beta)$ is satisfiable and $\theta_i \wedge \gamma
  (\beta)$ is satisfiable for some $i$. Since $\beta$ is canonical
  monomial, $\gamma(\beta)$ is a minterm over set $P$ and
  $\mathit{Lit}(\theta) \subseteq \mathit{Lit}(\gamma(\beta))$. Hence
  $\theta_i \wedge \gamma (\beta)$ is satisfiable implies $\gamma
  (\beta) \timplies \theta_i$. We have $\gamma (\beta) \timplies
  \theta$.

  The other direction is trivial.
\end{proof}

\begin{lem}
  \label{lemma:monotonic-abstraction}
  Let $P$ be a set of atomic predicates, $\theta, \rho \in
  \Prop[P]$. Then
  \begin{equation*}
    \theta \timplies \rho
    \mbox{ implies }
    \alpha (\theta) \timplies \alpha (\rho).
  \end{equation*}
\end{lem}

\begin{proof}
  Let $\alpha (\theta) = \bigvee\limits_i \beta_i$ where $\beta_i$ is
  a canonical monomial and $\theta \wedge \gamma (\beta_i)$ is
  satisfiable. By Lemma~\ref{lemma:exact}, 
  $\gamma (\beta_i) \timplies \theta$. Hence $\gamma (\beta_i)
  \timplies \rho$ and $\rho \wedge \gamma (\beta_i)$ is
  satisfiable.
\end{proof}

\begin{lem}
  \label{lemma:equivalent-concretization}
  Let $P$ be a set of atomic propositions and $\theta \in
  \Prop[P]$. Then $\theta \Leftrightarrow \gamma (\alpha
  (\theta))$.
\end{lem}
\begin{proof}
  Let $\theta' = \bigwedge\limits_i \theta_i$ be a quantified-free
  formula in disjunctive normal form such that $\theta'
  \Leftrightarrow \theta$. Let $\mu \in \Bool[B_P]$. Define
  \begin{equation*}
    \chi (\mu) = \bigwedge (\{ b_p : \mu (b_p) = \TT \} \cup
                           \{ \neg b_p : \mu (b_p) = \FF \}).
  \end{equation*}
  Note that $\chi (\mu)$ is a canonical monomial and $\mu \models \chi
  (\mu)$.

  Assume $\nu \models \theta$. Then $\nu \models \theta_i$ for some
  $i$. Consider the canonical monomial $\chi (\alpha^* (\nu))$. Note
  that $\nu \models \gamma (\chi (\alpha^* (\nu)))$. Thus $\chi
  (\alpha^* (\nu))$ is a disjunct in $\alpha (\theta)$. We have $\nu
  \models \gamma (\alpha (\theta))$.

  Conversely, assume $\nu \models \gamma (\alpha (\theta))$. Then $\nu
  \models \gamma (\beta)$ for some canonical monomial $\beta$ and
  $\gamma (\beta) \wedge \theta$ is satisfiable. By
  Lemma~\ref{lemma:exact}, $\gamma (\beta) \timplies \theta$. Hence
  $\nu \models \theta$.
\end{proof}

\begin{lem}
  \label{lemma:abstraction-correspondence}
  Let $P$ be a set of atomic propositions, $\theta \in
  \Prop[P]$, $\beta \in \Bool[B_P]$, and $\nu$ a
  valuation for $X$. Then
  \begin{enumerate}[\em(1)]
  \item $\nu \models \theta$ if and only if 
    $\alpha^* (\nu) \models \alpha(\theta)$; and
  \item $\nu \models \gamma (\beta)$ if and only if
    $\alpha^* (\nu) \models \beta$.
  \end{enumerate}
  \begin{equation*}
  \end{equation*}
\end{lem}
{
\begin{proof}\hfill
  \begin{enumerate}[(1)]
  \item Assume $\nu \models \theta$. $\chi (\alpha^* (\nu))$ is a
    canonical monomial. Observe that $\nu \models \gamma (\chi
    (\alpha^* (\nu)))$. Hence $\gamma (\chi (\alpha^* (\nu))) \wedge
    \theta$ is satisfiable. By the definition of $\alpha (\theta)$ and
    $\chi (\alpha^* (\nu))$ is canonical, $\chi (\alpha^* (\nu))
    \Rightarrow \alpha (\theta)$. $\alpha^* (\nu) \models \alpha
    (\theta)$ follows from $\alpha^* (\nu) \models \chi (\alpha^*
    (\nu))$.

    Conversely, assume $\alpha^* (\nu) \models \alpha (\theta)$. Then
    $\alpha^* (\nu) \models \beta$ where $\beta$ is a canonical
    monomial and $\gamma (\beta) \wedge \theta$ is satisfiable. By the
    definition of $\alpha^* (\nu)$, $\nu \models \gamma
    (\beta)$. Moreover, $\gamma (\beta) \Rightarrow \theta$ by
    Lemma~\ref{lemma:exact}. Hence $\nu \models \theta$.
  \item Assume $\nu \models \gamma (\beta)$. By
    Lemma~\ref{lemma:abstraction-correspondence}~1, $\alpha^* (\nu)
    \models \alpha (\gamma (\beta))$. Note that $\beta = \alpha
    (\gamma (\beta))$. Thus $\alpha^* (\nu) \models \beta$.
  \end{enumerate}
\end{proof}
}

\begin{lem}
  \label{lemma:concretization}
  Let $P$ be a set of atomic propositions, $\theta \in \Prop[P]$, and
  $\mu$ a Boolean valuation for $B_P$. Then $\gamma^* (\mu) \timplies
  \theta$ if and only if $\mu \models \alpha(\theta)$.
\end{lem}
{
\begin{proof}
  Assume $\gamma^* (\mu) \timplies \theta$. By
  Lemma~\ref{lemma:monotonic-abstraction}, $\alpha (\gamma^* (\mu))
  \timplies \alpha (\theta)$. Note that $\gamma^* (\mu) = \gamma
  (\chi (\mu))$. By Lemma~\ref{lemma:equivalent-concretization}, $\chi
  (\mu) \timplies \alpha (\theta)$. Since $\mu \models \chi (\mu)$,
  we have $\mu \models \alpha (\theta)$.

  Conversely, assume $\mu \models \alpha (\theta)$. We have $\chi
  (\mu) \timplies \alpha (\theta)$ by the definition of $\chi
  (\mu)$. Let $\nu \models \gamma^* (\mu)$, that is, $\nu \models
  \gamma (\chi (\mu))$. By
  Lemma~\ref{lemma:abstraction-correspondence}~(2), $\alpha^* (\nu)
  \models \chi (\mu)$. Since $\chi (\mu) \timplies \alpha (\theta)$,
  $\alpha^* (\nu) \models \alpha (\theta)$. By
  Lemma~\ref{lemma:abstraction-correspondence}~(1), $\nu \models
  \theta$. Therefore, $\gamma^* (\mu) \timplies \theta$. 
\end{proof}
}

\subsection{CDNF Learning Algorithm}
\label{subsec:cdnf}

CDNF algorithm~\cite{IC95} is an exact learning algorithm for Boolean
formulae based on monotone theory. It infers an unknown target formula
by posing queries to a teacher. The teacher is responsible for
answering two types of queries. The learning algorithm may ask if a
valuation satisfies the target formula by a membership query. Or it
may ask if a conjectured formula is equivalent to the target in an
equivalence query. Using the answers for the queries, CDNF algorithm
infers a Boolean formula equivalent to the unknown target within a
polynomial number of queries in the formula size of the
target~\cite{IC95}.

\subsection{Programs}
We consider the following imperative language in this paper:
\begin{equation*}
  \begin{array}{rcl}
    \mathsf{Stmt} & \defn &
    \mathtt{nop} \ |\ 
    \mathsf{Stmt} \mathtt{;}\ \mathsf{Stmt} \ |\ x \ \mathtt{:=}\
    \mathsf{Exp} \ |\  x \ \mathtt{:=}\  \mathtt{nondet} \ |\ 
   \mathtt{if}\ \mathsf{BExp} 
    \ \mathtt{then}\ 
    \mathsf{Stmt}
    \ \mathtt{else}\ 
    \mathsf{Stmt}\ 
    \\
	\mathsf{Exp} & \defn &
    n \ |\ x \ |\ 
    \mathsf{Exp} + \mathsf{Exp} \ |\ 
    \mathsf{Exp} - \mathsf{Exp}\\

	\mathsf{BExp} & \defn &
    \FF \ |\  x \ |\ 
    \neg \mathsf{BExp} \ |\ 
    \mathsf{BExp} \wedge \mathsf{BExp} \ |\ 
    \mathsf{Exp} < \mathsf{Exp} \ |\ 
    \mathsf{Exp} = \mathsf{Exp}
  \end{array}
\end{equation*}
Two basic types are available:
natural numbers and Booleans. A term in $\mathsf{Exp}$ is a natural
number; a term 
in $\mathsf{BExp}$ is of Boolean type.
The keyword $\mathtt{nondet}$ denotes an arbitrary value in the type
of the assigned variable. 
An \emph{annotated loop} is of the form:
\begin{equation*}
\mathtt{\{} \delta \mathtt{\}}\
\mathtt{while}\ \kappa\ \mathtt{do}\ S_1; S_2; \cdots; S_m\
\mathtt{done}\ \mathtt{\{} \epsilon \mathtt{\}}
\end{equation*}
The $\mathsf{BExp}$ formula $\kappa$ is the \emph{loop guard}.
The $\mathsf{BExp}$ formulae $\delta$ and $\epsilon$ are the
\emph{precondition} and \emph{postcondition} of the annotated loop
respectively. 

Define $\ox^{\langle k \rangle} = \{ x^{\langle k \rangle} : x \in \ox
\}$. For any term $e$ over $\ox$, define $e^{\langle k \rangle} = e[\ox
\mapsto \ox^{\langle k \rangle}]$.
A \emph{transition formula} $\valueof{S}$ for a statement $S$ is a
first-order formula over variables $\ox^{\langle 0 \rangle} \cup
\ox^{\langle 1 \rangle}$ defined as follows. 

\begin{equation*}
  \begin{array}{rcl}
    \valueof{\mathsf{nop}} & \defn & 
    \bigwedge\limits_{x \in \ox} x^{\langle 1 \rangle} = x^{\langle 0
      \rangle} \\
    \valueof{x := \mathsf{nondet}} & \defn &
    \bigwedge\limits_{y \in \ox \setminus \{ x \}} 
    y^{\langle 1 \rangle} = y^{\langle 0 \rangle}\\
    \valueof{x := e} & \defn &
    x^{\langle 1 \rangle} = e^{\langle 0 \rangle} \wedge 
    \bigwedge\limits_{y \in \ox \setminus \{ x \}} 
    y^{\langle 1 \rangle} = y^{\langle 0 \rangle}\\
    \valueof{S_0; S_1} & \defn & 
    \exists \ox. \valueof{S_0}[\ox^{\langle 1 \rangle} \mapsto \ox] \wedge
    \valueof{S_1}[\ox^{\langle 0 \rangle} \mapsto \ox]\\
    \valueof{\mathsf{if}\ p\ \mathsf{then}\ S_0\ \mathsf{else}\ S_1} &
    \defn &
    (p^{\langle 0 \rangle} \wedge \valueof{S_0}) \vee 
    (\neg p^{\langle 0 \rangle} \wedge \valueof{S_1})
  \end{array}
\end{equation*}
\\

Let $\nu$ and $\nu'$ be valuations, and $S$ a
statement. We write $\nu \goesto{S} \nu'$ if $\valueof{S}$
evaluates to true by assigning $\nu (x)$ and $\nu' (x)$ to
$x^{\langle 0 \rangle}$ and $x^{\langle 1 \rangle}$ for each $x \in \ox$
respectively. Given a sequence of  
statements $S_1; S_2; \cdots; S_m$, a \emph{program execution} 
$\nu_0 \goesto{S_1} \nu_1 \goesto{S_2} \cdots
\goesto{S_m} \nu_m$ is a sequence $[\nu_0, \nu_1, \ldots, \nu_m]$ of 
valuations such that $\nu_i \goesto{S_i} \nu_{i+1}$ for $0 \leq i < m$. 

A \emph{precondition} $\precond (\theta : S)$ for $\theta \in \Prop$
with respect to the statement $S$, which is a first-order formula that
entails $\theta$ after executing the statement $S$, is defined as
follows.

\begin{equation*}
  \begin{array}{rcl}
    \precond (\theta : \mathsf{nop}) & \defn & \theta\\
    \precond (\theta : x \ \mathsf{:=}\ \mathsf{nondet}) & \defn &
    \forall x. \theta \\
    \precond (\theta : x \ \mathsf{:=}\ e) & \defn &
    \theta[x \mapsto e] \\
    \precond (\theta : S_0\mathsf{;}\ S_1) & \defn &
    \precond (\precond (\theta : S_1) : S_0)\\
    \precond (\theta : \mathsf{if}\ p\ \mathsf{then}\ S_0\
    \mathsf{else}\ S_1) & \defn &
    (p \timplies \precond (\theta : S_0)) \wedge
    (\neg p \timplies \precond (\theta : S_1))\\
  \end{array}
\end{equation*}
\\
Observe that all universal quantifiers occur positively in $\precond
(\theta : S)$ for any $S$. They can be eliminated by Skolem
constants~\cite{POPL02,VMCAI04}.
\\

\subsection{Problem Definition}

Given an annotated loop,

\begin{equation*}
  \mathtt{\{} \delta \mathtt{\}}\
  \mathtt{while}\ \kappa\ \mathtt{do}\ S_1; S_2; \cdots; S_m\
  \mathtt{done}\ \mathtt{\{} \epsilon \mathtt{\}},
\end{equation*}
\\
the \emph{loop invariant inference problem} is to compute an invariant
$\iota \in \Prop$ that is a formula satisfying
\begin{enumerate}[(1)]
\item $\delta \Rightarrow \iota$;

\item $\iota \wedge \neg \kappa \Rightarrow \epsilon$; and

\item $\iota \wedge \kappa \Rightarrow \precond (\iota : S_1;
  S_2; \cdots; S_m)$.
\end{enumerate}
Observe that the condition (2) is equivalent to $\iota \Rightarrow
\epsilon \vee \kappa$. The first two conditions specify necessary and
sufficient conditions of any loop invariants respectively. The
formulae $\delta$ and $\epsilon \vee \kappa$ are called the
\emph{strongest} and \emph{weakest approximations} to loop invariants
respectively.

We are particularly interested in the following variant of the loop
invariant inference problem:

\begin{enumerate}[(a)]

\item Given a set $P$ of atomic predicates, finding an invariant
  $\iota \in \Prop[P]$; and

\item Given an annotated loop, finding a suitable set $P$ of
  atomic predicates that contains enough predicates to express at
  least one of the invariants.

\end{enumerate}

Jung \etal{} propose a algorithmic-learning-based
technique~\cite{VMCAI10} that solves the part (a) of the problem. The
technique combines predicate abstraction and decision procedures to
make a mechanical teacher that answers the queries from learning
algorithm. With predicate abstraction, the learning algorithm becomes
an efficient engine for exploring possible combinations of predicates
to find an invariant. 

In this paper, we address the part (b) of the problem using
interpolation. As already stated in
Section~\ref{subsec:interpolation}, interpolation provides a
systematic method for predicate generation and widely adopted in
software model checking. We explain the application of interpolation
in the context of learning-based loop invariant inference.

\section{Inferring Loop Invariants with Algorithmic Learning}
\label{section:inferring-unquantified-loop-invariants}

In this section, we review the learning-based framework for inferring
quantifier-free loop invariant due to Jung \etal~\cite{VMCAI10}. Given
a set $P$ of atomic predicates, the authors show how to apply a
learning algorithm for Boolean formulae to infer quantifier-free loop
invariants freely generated by $P$. They first adopt predicate
abstraction to relate quantifier-free and Boolean formulae. They then
design a mechanical teacher to guide the learning algorithm to a
Boolean formula whose concretization is a loop invariant. We first
explain the algorithms for resolving queries from the learning
algorithm and then the main loop of learning-based loop invariant
inference.

\subsection{Answering Queries from Algorithmic Learning}
\label{subsec:answering_queries}

\begin{figure}
  \centering
  \resizebox{0.6\textwidth}{!}{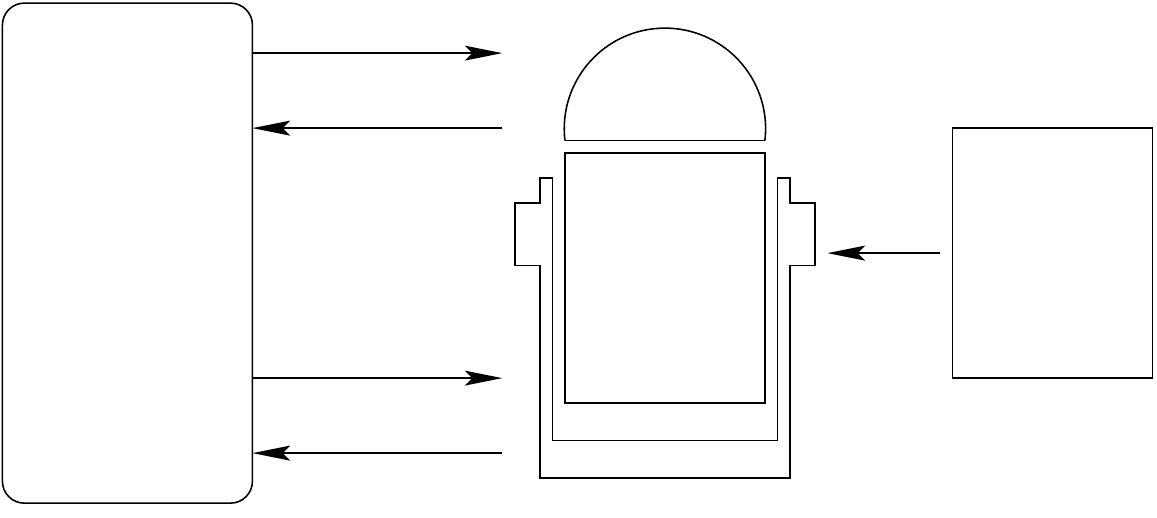}
  \caption{Learning-based Framework}
  \label{figure:framework}
\end{figure}

Figure~\ref{figure:framework} shows a high-level view of
learning-based loop invariant inference framework. In the framework, a
learning algorithm is used to drive the search of loop invariants. It
``learns'' an unknown loop invariant by inquiring a mechanical
teacher. The mechanical teacher of course does not know any loop
invariant. It nevertheless tries to answer these queries by the
information derived from program texts. In this case, the teacher uses
approximations to loop invariants. By employing a learning algorithm,
it suffices to design a mechanical teacher to find loop
invariants. Moreover, the new framework does not construct abstract
models nor compute fixed points. It can be more scalable than
traditional techniques.

After formulae in $\Prop$ and valuations in $\Val{\ox}$ are abstracted
to those in $\Bool[\obp]$ and $\Val{\obp}$ respectively, a learning
algorithm is used to infer abstractions of loop invariants. Let $\xi$
be an unknown \emph{target} Boolean formula in $\Bool[\obp]$. A
learning algorithm computes a representation of the target
$\xi$ by interacting with a teacher. The \emph{teacher} should answer
the following queries~\cite{IC95}:
\begin{iteMize}{$\bullet$}
\item \emph{Membership queries}. Let $\mu \in \Val{\obp}$ be an
  abstract valuation. The membership
  query $\MEM (\mu)$ asks if the unknown target $\xi$ is satisfied by
  $\mu$. If so, the teacher answers $\YES$; otherwise, $\NO$.
\item \emph{Equivalence queries}. Let $\beta \in
  \Bool[\obp]$ be an
  \emph{abstract conjecture}. The equivalence query $\EQ (\beta)$ asks if
  $\beta$ is equivalent to the unknown target $\xi$. If so, the
  teacher answers $\YES$. Otherwise, the teacher gives an abstract
  valuation $\mu$ such that the exclusive disjunction of $\beta$ and
  $\xi$ is satisfied by $\mu$. The abstract valuation $\mu$ is called an
  \emph{abstract counterexample}. 
\end{iteMize}

With predicate abstraction and a learning algorithm for Boolean
formulae at hand, it remains to design a mechanical teacher to guide the
learning algorithm to the abstraction of a loop invariant.
The key idea in~\cite{VMCAI10} is to exploit approximations to loop
invariants. An \emph{under-approximation} to loop invariants is a
quantifier-free formula $\underline{\iota}$ which is stronger than
some loop invariants of the given annotated loop; an
\emph{over-approximation} is a quantifier-free
formula $\overline{\iota}$ which is weaker than some loop
invariants. 

In the following, we explain exactly how we can answer queries from
learning algorithm using under- and over-approximation of loop
invariant.

\subsubsection{Answering Membership Queries}
In the membership query $\MEM (\mu)$, the teacher 
is required to answer whether $\mu \models \alpha (\xi)$. We
concretize the Boolean valuation $\mu$ and check it against the
approximations.  If the concretization $\gamma^* (\mu)$ is
inconsistent (that is, $\gamma^* (\mu)$ is unsatisfiable), we simply
answer $\mathit{NO}$ for the membership query. Otherwise, there are
three cases:
\begin{enumerate}[(1)]
\item $\gamma^* (\mu) \timplies \underline{\iota}$. Thus $\mu \models
  \alpha (\underline{\iota})$ (Lemma~\ref{lemma:concretization}). 
  And $\mu \models \alpha (\iota)$ by
  Lemma~\ref{lemma:monotonic-abstraction}.
\item $\gamma^* (\mu) \nimplies \overline{\iota}$. Thus $\mu
  \not\models \alpha (\overline{\iota})$
  (Lemma~\ref{lemma:concretization}). That is, $\mu \models \neg
  \alpha (\overline{\iota})$. Since $\iota \rightarrow 
  \overline{\iota}$, we have $\mu \not\models \alpha (\iota)$ by
  Lemma~\ref{lemma:monotonic-abstraction}.
\item Otherwise, we cannot determine whether $\mu \models \alpha
  (\iota)$ by the approximations. In this case, we answer YES or NO
  randomly.
\end{enumerate}  

\begin{algorithm}
  \tcc{$\underline{\iota}, \overline{\iota}$ : under- and
    over-approximations to loop invariants}
  \KwIn{a membership query $\MEM (\mu)$ with $\mu \in \Val{\obp}$}
  \KwOut{$\YES$ or $\NO$}
  $\theta$ := $\gamma^* (\mu)$\;
  \lIf{$\theta$ is inconsistent}
  {
    \Return{$\NO$}\;
  }
  \lIf{$\theta \Rightarrow \underline{\iota}$}
  {
    \Return{$\YES$}\;
  }
  \lIf{$\nu \models \neg (\theta \Rightarrow \overline{\iota})$}
  {
    \Return{$\NO$}\;
  }
  \Return{$\YES$ or $\NO$ randomly}\;
  \phantom{x}
  \caption{Membership Query Resolution}
  \label{algorithm:membership}
\end{algorithm}

Algorithm~\ref{algorithm:membership} shows our membership query
resolution algorithm. Note that instead of giving a random answer when
a membership query cannot be resolved by given invariant
approximations, one can give more accurate answer by exploiting better
approximations from static analyzers. This learning-based framework is
orthogonal to existing static analysis techniques~\cite{VMCAI10}.

\subsubsection{Answering Equivalence Queries}
To answer the equivalence query $\EQ (\beta)$, we concretize the
Boolean formula $\beta$ and check if $\gamma (\beta)$ is indeed an
invariant of the \texttt{while} statement for the given pre- and
post-conditions. If it is, we are done. Otherwise, we use an SMT
solver to find a witness to $\alpha (\xi) \oplus \beta$.  There are
three cases:
\begin{enumerate}[(1)]
\item There is a $\nu$ such that $\nu \models \neg (\underline{\iota}
  \timplies \gamma (\beta))$. Then $\nu \models \underline{\iota}
  \wedge \neg \gamma (\beta)$. By Lemma~\ref{lemma:abstraction-correspondence}
  and~\ref{lemma:monotonic-abstraction}, we have $\alpha^* (\nu)
  \models \alpha (\iota)$ and $\alpha^* (\nu) \models \neg \beta$.
  Thus, $\alpha^* (\nu)
  \models \alpha (\xi) \wedge \neg \beta$.
\item There is a $\nu$ such that $\nu \models \neg (\gamma (\beta)
  \timplies \overline{\iota})$. Then $\nu \models \gamma (\beta)
  \wedge \neg \overline{\iota}$. By
  Lemma~\ref{lemma:abstraction-correspondence}, 
  $\alpha^* (\nu) \models \beta$. $\alpha^* (\nu) \models \neg \alpha
  (\iota)$ by Lemma~\ref{lemma:abstraction-correspondence}
  and~\ref{lemma:monotonic-abstraction}. Hence $\alpha^* (\nu) \models
  \beta \wedge \neg \alpha (\xi)$.
\item Otherwise, we cannot find a witness to $\alpha (\xi) \oplus
  \beta$ by the approximations. In this case, we give a random
  abstract counterexample.
\end{enumerate}

\begin{algorithm}
  \tcc{$\mathtt{\{} \delta \mathtt{\}}\ \mathtt{while}\ \kappa\
    \mathtt{do}\ S_1; S_2; \cdots; S_m\ \mathtt{done}\
    \mathtt{\{} \epsilon \mathtt{\}}$ : an annotated loop}
  \tcc{$\underline{\iota}, \overline{\iota}$ : under- and
    over-approximations to loop invariants}
  \KwIn{an equivalence query $\EQ (\beta)$ with $\beta \in
    \Bool[\obp]$} 
  \KwOut{$\YES$ or an abstract counterexample}
  $\theta$ := $\gamma (\beta)$\;
  \lIf{$\delta \Rightarrow \theta$ and $\theta \Rightarrow \epsilon
    \vee \kappa$ and $\theta \wedge \kappa \Rightarrow \precond
    (\theta : S_1; S_2; \cdots; S_m)$}
  \Return{$\YES$}\;
  \uIf{$\nu \models \neg (\underline{\iota} \Rightarrow \theta)$ or
    $\nu \models \neg (\theta \Rightarrow \overline{\iota})$ or
    $\nu \models \neg (\theta \wedge \kappa \Rightarrow \precond
    (\overline{\iota} : S_1; S_2; \cdots; S_m))$}
  { \Return{$\alpha^* (\nu)$}\;}
  \Return{a random abstract counterexample}\;
  \phantom{x}
  \caption{Equivalence Query Resolution}
  \label{algorithm:equivalence}
\end{algorithm}

Algorithm~\ref{algorithm:equivalence} shows our equivalence query
resolution algorithm. Note that Algorithm~\ref{algorithm:equivalence}
returns $\mathit{YES}$ only if an invariant is found.

As in the membership query resolution, we give a random answer when an
equivalence query is not resolved by given invariant approximations. We
can still refine approximations using some static analysis to give
more accurate counterexample.

\subsection{Main Loop of of Inference Framework}
\label{subsec:main-loop}

\begin{algorithm}[h]
  \tcc{$\mathtt{\{} \delta \mathtt{\}}\ \mathtt{while}\ \kappa\
    \mathtt{do}\ S_1; S_2; \cdots; S_m\ \mathtt{done}\
    \mathtt{\{} \epsilon \mathtt{\}}$ : an annotated loop}
  \KwOut{a loop invariant for the annotated loop}
  $\underline{\iota}$ := $\delta \vee \epsilon$\;
  $\overline{\iota}$ := $\epsilon \vee \kappa$\;
  \Repeat{a loop invariant is found}
  {
    \ \ \textbf{call} a learning algorithm for Boolean formulae
    where membership and\\
    \ \ equivalence queries are resolved by
    Algorithms~\ref{algorithm:membership} 
    and~\ref{algorithm:equivalence} respectively\;
  }
  \phantom{x}
  \caption{Main Loop}
  \label{algorithm:main-loop}
\end{algorithm}

The main loop of loop invariant inference algorithm is given in
Algorithm~\ref{algorithm:main-loop}. We heuristically choose $\delta
\vee \epsilon$ and $\epsilon \vee \kappa$ as the under- and
over-approximations respectively. Note that the under-approximation
would be stronger if one uses the strongest approximation $\delta$. It
is, however, reported that the weaker approximation $\delta \vee
\epsilon$ for the under-approximation is more effective in resolving
queries~\cite{VMCAI10}. After determining the approximations, a
learning algorithm is used to find an invariant. In~\cite{VMCAI10},
Jung \etal{} use CDNF algorithm with
Algorithms~\ref{algorithm:membership} and~\ref{algorithm:equivalence}
for resolving queries.

Note that the mechanical teacher may give conflicting answers. Random
answers to membership queries may contradict abstract counterexamples
from equivalence queries. Moreover, different valuations may
correspond to the same abstract valuation. The learning algorithm
cannot infer any loop invariant in the presence of conflicting
answers. When the mechanical teacher gives conflicting answers, we
restart the learning algorithm and search another loop invariant. In
practice, there are nevertheless sufficiently many invariants for an
annotated loop. The learning-based technique can infer a loop
invariant without incurring any conflicts after a small number of
restarts. As an empirical evidence, observe the number of restarts in
Table~\ref{table:experiments}. Even without the new predicate
generation technique, the numbers of restarts in all but three
examples are less than three. The number of restarts is dramatically
improved with the new technique since the technique generates
predicates incrementally on demand so that it can make the abstraction
parsimonious.

We remark that the learning-based loop invariant inference is
semi-algorithm; Algorithm~\ref{algorithm:main-loop} terminates with a
loop invariant only when there exists one for the loop that can be
expressed with the given set of predicates. If there are not enough
atomic predicates to express any invariant, the algorithm will iterate
indefinitely. For example, \texttt{tar} example in
Section~\ref{section:experimental-results} timed out because it turned
out to have no invariant with only atomic predicates from the program
text.

\section{Predicate Generation by Interpolation}
\label{section:predicate-generation-by-interpolation}

One drawback in the learning-based approach to loop invariant
inference is to require a set of atomic predicates. 
It is essential that at least one quantifier-free loop
invariant is representable by the given set $P$ of atomic predicates.
Otherwise, concretization of formulae in $\Bool[\obp]$ cannot be
loop invariants. The mechanical teacher never answers $\YES$ to 
equivalence queries. To address this problem, we will synthesize new
atomic predicates for the learning-based loop invariant inference framework
progressively. 

The interpolation is essential to our predicate generation
technique. Let $\Theta = [ \theta_1, \theta_2, \ldots, \theta_m ]$ be
an inconsistent sequence of quantifier-free formula and $\Lambda =
[\lambda_0, \lambda_1, \lambda_2, \ldots, \lambda_m]$ its inductive
interpolant. By definition, $\theta_1 \Rightarrow \lambda_1$. Assume
$\theta_1 \wedge \theta_2 \wedge \cdots \wedge \theta_i \Rightarrow
\lambda_i$. We have $\theta_1 \wedge \theta_2 \wedge \cdots \wedge
\theta_{i+1} \Rightarrow \lambda_{i+1}$ since $\lambda_i \wedge
\theta_{i+1} \Rightarrow \lambda_{i+1}$. Thus, $\lambda_i$ is an
over-approximation to $\theta_1 \wedge \theta_2 \wedge \cdots \wedge
\theta_i$ for $0 \leq i \leq m$. Moreover, $\symbols{\lambda_i}
\subseteq \symbols{\theta_i} \cap \symbols{\theta_{i+1}}$. Hence
$\lambda_i$ can be seen as a concise summary of $\theta_1 \wedge
\theta_2 \wedge \cdots \wedge \theta_i$ with restricted symbols. Since
each $\lambda_i$ is written in a less expressive vocabulary, new
atomic predicates among variables can be synthesized. We therefore
apply the interpolation theorem to synthesize new atomic predicates
and refine the abstraction.

Our predicate generation technique consists of three components. 
Before the learning algorithm is invoked, an initial set of atomic
predicates is computed (Section~\ref{section:initial-predicates}). 
When the learning algorithm is failing to infer loop invariants, new
atomic predicates are generated to refine the abstraction
(Section~\ref{section:predicates-incorrect-conjectures}). Lastly,
conflicting 
answers to queries may incur from predicate abstraction. We 
further refine the abstraction with these conflicting answers
(Section~\ref{section:predicates-conflicting-answers}).
Throughout this section, we consider the annotated loop $\mathtt{\{}
\delta \mathtt{\}}\ \mathtt{while}\ \kappa\ \mathtt{do}\ S_1; S_2;$
$\cdots;$ $S_m\ \mathtt{done}\ \mathtt{\{} \epsilon \mathtt{\}}$ with the
under-approximation $\underline{\iota}$ and over-approximation
$\overline{\iota}$.

\subsection{Initial Atomic Predicates}
\label{section:initial-predicates}

The under- and over-approximations to loop invariants must satisfy
$\underline{\iota} \Rightarrow \overline{\iota}$. Otherwise, there
cannot be any loop invariant $\iota$ such that $\underline{\iota}
\Rightarrow \iota$ and $\iota \Rightarrow \overline{\iota}$. Thus,
the sequence
$[\underline{\iota}, \neg \overline{\iota}]$ is inconsistent.
For any interpolant
$[\TT, \lambda, \FF]$ of $[\underline{\iota}, \neg \overline{\iota}]$,
we have $\underline{\iota} \Rightarrow \lambda$ and $\lambda \Rightarrow
\overline{\iota}$. The quantifier-free formula $\lambda$ can be a loop
invariant if it satisfies $\lambda \wedge \kappa \Rightarrow \precond
(\lambda : S_1; S_2; \cdots; S_m)$. It is however unlikely that
$\lambda$ happens to be a loop invariant. Yet our loop invariant
inference algorithm can generalize $\lambda$ by taking the atomic
predicates in $\lambda$ as the initial atomic predicates. The 
learning algorithm will try to infer a loop invariant freely generated by
these atomic predicates. 

\subsection{Atomic Predicates from Incorrect Conjectures}
\label{section:predicates-incorrect-conjectures}

Consider an equivalence query $\EQ (\beta)$ where $\beta \in
\Bool[\obp]$ is an abstract conjecture. If the concretization $\theta
= \gamma (\beta)$ is not a loop invariant, we interpolate the loop
body with the incorrect conjecture $\theta$. For any
quantifier-free formula $\theta$ over variables $\ox^{\langle 0
  \rangle} \cup \ox^{\langle 1 \rangle}$, define 
$\theta^{\langle k \rangle} = 
\theta[\ox^{\langle 0 \rangle} \mapsto \ox^{\langle k \rangle}, 
\ox^{\langle 1 \rangle} \mapsto \ox^{\langle k+1 \rangle}]$. 
The \emph{desuperscripted} form of 
a quantifier-free formula $\lambda$ over variables 
$\ox^{\langle k \rangle}$ is $\lambda[\ox^{\langle k \rangle}
\mapsto \ox]$. Moreover, if $\nu$ is a valuation over 
$\ox^{\langle 0 \rangle} \cup \cdots \cup \ox^{\langle m \rangle}$,
$\nu\restrict{\ox^{\langle k \rangle}}$ represents a valuation over $\ox$
  such that $\nu\restrict{\ox^{\langle k \rangle}} (x) = \nu
  (x^{\langle k \rangle})$ for $x \in \ox$.
Let $\phi$ and $\psi$ be quantifier-free formulae over
$\ox$. Define the following sequence:
\begin{equation*}
  \Xi (\phi, S_1, \ldots, S_m, \psi) \defn
  [\phi^{\langle 0 \rangle}, \valueof{S_1}^{\langle 0 \rangle},
  \valueof{S_2}^{\langle 1 \rangle}, \ldots, 
  \valueof{S_m}^{\langle m - 1 \rangle}, \neg \psi^{\langle m \rangle}].
\end{equation*}

Observe that 
\begin{iteMize}{$\bullet$}
\item $\phi^{\langle 0 \rangle}$ and $\valueof{S_1}^{\langle 0
    \rangle}$ share the variables $\ox^{\langle 0 \rangle}$; 
\item $\valueof{S_m}^{\langle m - 1 \rangle}$ and $\neg \psi^{\langle
    m \rangle}$ share the variables $\ox^{\langle m \rangle}$; and
\item $\valueof{S_i}^{\langle i - 1 \rangle}$ and
  $\valueof{S_{i+1}}^{\langle i \rangle}$ share the variables
    $\ox^{\langle i \rangle}$ for $1 \leq i < m$.
\end{iteMize}
Starting from the program states satisfying 
$\phi^{\langle 0 \rangle}$, the formula
\begin{equation*}
  \phi^{\langle 0 \rangle} \wedge \valueof{S_1}^{\langle 0 \rangle}
  \wedge \valueof{S_2}^{\langle 1 \rangle} \wedge \cdots
  \wedge \valueof{S_i}^{\langle i - 1 \rangle}
\end{equation*}
characterizes the images of $\phi^{\langle 0 \rangle}$ during the
execution of $S_1; S_2; \cdots; S_i$. 

\begin{lem}
  Let $\ox$ denote the set of variables in the statement $S_1;
  S_2; \cdots; S_i$, and $\phi$ a quantifier-free formula over $\ox$. 
  For any valuation $\nu$ over $\ox^{\langle 0 \rangle} \cup
  \ox^{\langle 1 \rangle} \cup \cdots \cup \ox^{\langle i \rangle}$, the
  formula $\phi^{\langle 0 \rangle} \wedge 
  \valueof{S_1}^{\langle 0 \rangle} \wedge 
  \valueof{S_2}^{\langle 1 \rangle} \wedge \cdots \wedge
  \valueof{S_i}^{\langle i - 1 \rangle}$ is satisfied by $\nu$ if and
  only if
  $\nu\restrict{\ox^{\langle 0 \rangle}} \goesto{S_1} 
  \nu\restrict{\ox^{\langle 1 \rangle}} \goesto{S_2} \cdots \goesto{S_i}
  \nu\restrict{\ox^{\langle i \rangle}}$ is a program execution
  and $\nu\restrict{\ox^{\langle 0 \rangle}} \models \phi$.
\end{lem}
\begin{proof}
  By induction on the length of statement
  $S_1;S_2;\cdots;S_i$. Suppose that the lemma is true for statement
  $S_1;S_2;\cdots;S_i$. By definition of program execution, if
  $\nu\restrict{\ox^{\langle i \rangle}} \goesto{S_{i+1}}
  \nu\restrict{\ox^{\langle i + 1 \rangle}}$, then $\nu$ satisfies
  $\valueof{S_{i + 1}}^{\langle i \rangle}$ and vice versa. By
  induction hypothesis, the formula $\phi^{\langle 0 \rangle} \wedge
  \valueof{S_1}^{\langle 0 \rangle} \wedge \valueof{S_2}^{\langle 1
    \rangle} \wedge \cdots \wedge \valueof{S_{i + 1}}^{\langle i
    \rangle}$ is satisfied by $\nu$ and the statement follows by it.
\end{proof}

By definition, $\phi
\Rightarrow \precond (\psi : S_1; S_2; \cdots; S_m)$ implies that the image
of $\phi$ must satisfy $\psi$ after the execution of $S_1; S_2;
\cdots; S_m$. The sequence $\Xi (\phi, S_1, \ldots, S_m, \psi)$
is inconsistent if $\phi \Rightarrow \precond (\psi : S_1; S_2; \cdots;
S_m)$. The following proposition will be handy.

\begin{prop}
  \label{proposition:approximation}
  Let $S_1; S_2; \cdots; S_m$ be a sequence of statements.
  For any $\phi$ with $\phi \Rightarrow \precond (\psi :
  S_1; S_2; \cdots; S_m)$, $\Xi (\phi, S_1, \ldots, S_m, \psi)$ has
  an inductive interpolant.
\end{prop}

\begin{proof}
  By induction on the length of statement
  $S_1;S_2;\cdots;S_m$. Suppose the proposition holds for statement
  $S_2;\cdots;S_m$ and an arbitrary formula $\phi$ with $\phi
  \Rightarrow \precond (\psi : S_1; S_2; \cdots; S_m)$. By definition
  of $\precond$, $\precond (\psi : S_1; S_2; \cdots; S_m) = \precond
  (\precond(\psi : S_2; \cdots; S_m) : S_1)$ Let $\phi'$ be a formula
  such that $\phi$ satisfies $\phi'$ after execution of $S_1$. By
  induction hypothesis, $\Xi(\phi', S_2, \ldots, S_m, \psi)$ has an
  inductive interpolant. Thus, $\Xi(\phi, S_1, \ldots, S_m, \psi)$
  also has inductive interpolant.
  
\end{proof}

Let $\Lambda = [\TT, \lambda_1, \lambda_2, \ldots, \lambda_{m+1}, \FF]$ be
an inductive interpolant of $\Xi (\phi, S_1, \ldots, S_m, \psi)$. Recall that
$\lambda_i$ is a quantifier-free formula over 
$\ox^{\langle i - 1 \rangle}$ for $1 \leq i \leq m + 1$. It is also an
over-approximation to the 
image of $\phi$ after executing $S_1; S_2; \cdots; S_{i - 1}$. 
Proposition~\ref{proposition:approximation} can be used to generate new
atomic predicates. One simply finds a pair of quantifier-free formulae
$\phi$ and $\psi$ with $\phi \Rightarrow \precond (\psi : S_1; S_2;
\cdots; S_m)$, applies the interpolation theorem, and collects
desuperscripted atomic predicates in an inductive interpolant of $\Xi (\phi,
S_1, \ldots, S_m, \psi)$. In the
following, we show how to obtain such pairs with under- and
over-approximations to loop invariants.

\subsubsection{Interpolating Over-Approximation}

It is not hard to see that an over-approximation to loop invariants
characterizes loop invariants after the execution of the loop
body. Recall that $\iota \Rightarrow \overline{\iota}$ for some loop
invariant $\iota$. Moreover, $\iota \wedge \kappa \Rightarrow
\precond (\iota : S_1; S_2; \cdots; S_m)$. By
the monotonicity of $\precond (\bullet : S_1; S_2; \cdots; S_m)$,
we have $\iota \wedge \kappa
\Rightarrow \precond (\overline{\iota} : S_1; S_2; \cdots; S_m)$. 

\begin{prop}
  \label{proposition:necessary-over-approximation}
  Let $\overline{\iota}$ be an over-approximation to loop invariants
  of the annotated loop $\mathtt{\{} \delta \mathtt{\}}\
  \mathtt{while}\ \kappa\  
  \mathtt{do}\ S_1; S_2; \cdots; S_m\ \mathtt{done}\ \mathtt{\{}
  \epsilon \mathtt{\}}$. 
  For any loop invariant $\iota$ with $\iota \Rightarrow
  \overline{\iota}$, $\iota \wedge \kappa \Rightarrow \precond
  (\overline{\iota} : S_1; S_2; \cdots; S_m)$. 
\end{prop}
\zap{
\begin{proof}
  Since $\iota$ is a loop invariant, $\iota \wedge \kappa \Rightarrow
  \precond (\iota : S)$. The statement follows by the monotonicity of
  $\precond (\bullet : S)$.
\end{proof}
}

Proposition~\ref{proposition:necessary-over-approximation} gives a
necessary condition to loop invariants of interest. Recall that
$\theta = \gamma (\beta)$ is an incorrect conjecture of loop
invariants. If $\nu \models \neg (\theta \wedge \kappa \Rightarrow
\precond (\overline{\iota} : S_1; S_2; \cdots; S_m))$, the mechanical
teacher returns the abstract counterexample $\alpha^* (\nu)$. Otherwise,
Proposition~\ref{proposition:approximation} is applicable with the
pair $\theta \wedge \kappa$ and $\overline{\iota}$.

\begin{cor}
  \label{corollary:over-approximation}
  Let $\overline{\iota}$ be an
  over-approximation to loop invariants of the annotated loop
  $\mathtt{\{} \delta \mathtt{\}}\ \mathtt{while}\ \kappa\ 
  \mathtt{do}\ S_1; S_2; \cdots; S_m\ \mathtt{done}\ \mathtt{\{}
  \epsilon \mathtt{\}}$. 
  For any $\theta$ with $\theta \wedge \kappa \Rightarrow \precond
  (\overline{\iota} : S_1; S_2; \cdots; S_m)$, the sequence $\Xi
  (\theta \wedge \kappa, S_1, S_2, \ldots, S_m, \overline{\iota})$ has
  an inductive interpolant.
\end{cor}
\begin{proof}
  By Proposition~\ref{proposition:approximation}.
\end{proof}

\subsubsection{Interpolating Under-Approximation}

For under-approximations, there is no necessary
condition. Nevertheless, Proposition~\ref{proposition:approximation}
is applicable with the pair $\underline{\iota} \wedge \kappa$
and $\theta$.

\begin{cor}
  \label{corollary:under-approximation}
  Let $\underline{\iota}$ be an under-approximation to loop invariants
  of the annotated loop $\mathtt{\{} \delta \mathtt{\}}\
  \mathtt{while}\ \kappa\ \mathtt{do}\ S_1; S_2; \cdots; S_m\
  \mathtt{done}\ \mathtt{\{} \epsilon \mathtt{\}}$.  For any $\theta$
  with $\underline{\iota} \wedge \kappa \Rightarrow \precond (\theta :
  S_1; S_2; \cdots; S_m)$, the sequence $\Xi (\underline{\iota} \wedge
  \kappa, S_1, S_2, \ldots, S_m, \theta)$ has an inductive
  interpolant.
\end{cor}
\begin{proof}
  By Proposition~\ref{proposition:approximation}.
\end{proof}

Generating atomic predicates from an incorrect conjecture $\theta$
should now be clear
(Algorithm~\ref{algorithm:generating-predicates}). 
Assuming that the incorrect conjecture satisfies the necessary
condition in
Proposition~\ref{proposition:necessary-over-approximation}, we simply
collect all desuperscripted atomic predicates 
in an inductive interpolant of $\Xi (\theta \wedge \kappa, S_1, S_2,
\ldots, S_m, \overline{\iota})$
(Corollary~\ref{corollary:over-approximation}). 
More atomic predicates can be obtained from an inductive
interpolant
of $\Xi (\underline{\iota} \wedge \kappa, S_1, S_2, \ldots, S_m,
\theta)$ if additionally
$\underline{\iota} \wedge \kappa \Rightarrow \precond (\theta : S_1;
S_2; \cdots; S_m)$
(Corollary~\ref{corollary:under-approximation}).

\begin{algorithm}
  \tcc{$\mathtt{\{} \delta \mathtt{\}}\ \mathtt{while}\ \kappa\
    \mathtt{do}\ S_1\mathtt{;} \cdots \mathtt{;} S_m\ \mathtt{done}\
    \mathtt{\{} \epsilon \mathtt{\}}$ : an annotated loop}
  \tcc{$\underline{\iota}, \overline{\iota}$ : under- and
    over-approximations to loop invariants}
  \KwIn{a formula $\theta \in \Prop[P]$ such that $\theta \wedge
    \kappa \Rightarrow \precond (\overline{\iota} : S_1; S_2; \cdots; S_m)$}
  \KwOut{a set of atomic predicates}

  $I$ := an inductive interpolant of $\Xi (\theta 
  \wedge \kappa, S_1, S_2, \ldots, S_m, \overline{\iota})$\;
  $Q$ := desuperscripted atomic predicates in $I$\;
  \If{$\underline{\iota} \wedge \kappa \Rightarrow \precond (\theta :
    S_1; S_2; \cdots; S_m)$}
  {
    $J$ := an inductive interpolants of $\Xi
    (\underline{\iota} \wedge \kappa, S_1, S_2, \ldots, S_m,
    \theta)$\;
    $R$ := desuperscripted atomic predicates in $J$\;
    $Q$ := $Q \cup R$\;
  }
  \Return {$Q$}
  \phantom{x}
  \caption{$\mathtt{PredicatesFromConjecture}(\theta)$}
  \label{algorithm:generating-predicates}
\end{algorithm}

\subsection{Atomic Predicates from Conflicting Abstract
  Counterexamples} 
\label{section:predicates-conflicting-answers}

Because of the abstraction, conflicting abstract counterexamples may
be given to the learning algorithm.  Consider the example in
Section~\ref{section:introduction}. Recall that $n \geq 0 \wedge x = n
\wedge y = n$ and $x + y = 0 \vee x > 0$ are the under- and
over-approximations respectively. Suppose there is only one atomic
predicate $y = 0$. The learning algorithm tries to infer a Boolean
formula $\lambda \in \Bool[b_{y = 0}]$.  Let us resolve
the equivalence queries $\EQ (\TT)$ and $\EQ (\FF)$. On the
equivalence query $\EQ (\FF)$, we check if $\FF$ is weaker than the
under-approximation by an SMT solver. It is not, and the SMT solver
gives the valuation $\nu_0 (n) = \nu_0 (x) = \nu_0 (y) = 1$ as a
witness. Applying the abstraction function $\alpha^*$ to $\nu_0$, the
mechanical teacher returns the abstract counterexample 
$b_{y = 0} \mapsto \FF$. The abstract counterexample is
intended to notify that the target formula $\lambda$ and
$\FF$ have different truth values when $b_{y = 0}$ is $\FF$. That is,
$\lambda$ is satisfied by the valuation $b_{y=0} \mapsto \FF$.

On the equivalence query $\EQ (\TT)$, the mechanical teacher checks if
$\TT$ is stronger than the over-approximation. It is not, and the SMT
solver now 
returns the valuation $\nu_1 (x) = 0, \nu_1 (y) = 1$ as a witness.
The mechanical teacher in turn computes 
$b_{y = 0} \mapsto \FF$ as the corresponding abstract
counterexample. The abstract counterexample notifies
that the target formula $\lambda$ and $\TT$ have
different truth values when $b_{y = 0}$ is $\FF$. That
is, $\lambda$ is not satisfied by the valuation $b_{y=0} \mapsto \FF$. 
Yet the target formula $\lambda$ cannot be satisfied and
unsatisfied by the valuation $b_{y = 0} \mapsto \FF$. 
We have conflicting abstract counterexamples.

Such conflicting abstract counterexamples arise because the
abstraction is too coarse. This gives us another chance to refine the
abstraction. For distinct valuations $\nu$ and $\nu'$, $\Gamma (\nu)
\wedge \Gamma (\nu')$ is inconsistent. For instance, $\Gamma (\nu_0) =
(n = 1) \wedge (x = 1) \wedge (y = 1)$, $\Gamma (\nu_1) = (x = 0)
\wedge (y = 1)$, and $\Gamma (\nu_1) \wedge \Gamma (\nu_0)$ is
inconsistent.

\begin{algorithm}
\zap{
  \tcc{$\mathtt{\{} \delta \mathtt{\}}\ \mathtt{while}\ \kappa\
    \mathtt{do}\ S_1; S_2; \cdots; S_m\ \mathtt{done}\
    \mathtt{\{} \epsilon \mathtt{\}}$ : an annotated loop}
}
  \KwIn{distinct valuations $\nu$ and $\nu'$ such that $\alpha^* (\nu)
    = \alpha^* (\nu')$}
  \KwOut{a set of atomic predicates}

  $X$ := $\Gamma (\nu)$\;
  $X'$ := $\Gamma (\nu')$\;
  \tcc{$X \wedge X'$ is inconsistent}
  $\rho$ := $\gamma^* (\alpha^* (\nu))$\;
  $Q$ := atomic predicates in an inductive interpolant of
  $[X, X' \vee \neg \rho]$\;
  \Return{$Q$}\;
  \phantom{x}
  \caption{$\mathtt{PredicatesFromConflict}(\nu, \nu')$}
  \label{algorithm:generating-predicate-conflicting-counterexamples}
\end{algorithm}

Algorithm~\ref{algorithm:generating-predicate-conflicting-counterexamples}
generates atomic predicates from conflicting abstract counterexamples.
Let $\nu$ and $\nu'$ be distinct valuations in $\Val{\ox}$. We compute
formulae $X = \Gamma (\nu)$ and $X' = \Gamma (\nu')$. Since $\nu$ and
$\nu'$ are conflicting, they correspond to the same abstract valuation
$\alpha^* (\nu) = \alpha^* (\nu')$. Let $\rho = \gamma^* (\alpha^*
(\nu))$. We have $X \Rightarrow \rho$ and $X' \Rightarrow
\rho$~\cite{VMCAI10}.  Recall that $X \wedge X'$ is inconsistent.
$[X, X' \vee \neg \rho]$ is also inconsistent for $X \Rightarrow
\rho$.
Algorithm~\ref{algorithm:generating-predicate-conflicting-counterexamples}
returns atomic predicates in an inductive interpolant of $[X, X' \vee
\neg \rho]$.

\section{Loop Invariant Inference Algorithms with Predicate
  Generation}
\label{section:algorithm}

Algorithm~\ref{algorithm:main-loop-new} is the main loop of inference
framework with predicate generation. The algorithm is the same as
Algorithm~\ref{algorithm:main-loop} except the gray-boxed parts.
\begin{algorithm}[h]
  \tcc{$\CEX$ : a set of counterexamples}
  \tcc{$\tau$ : a threshold to generate new atomic predicates}
  \tcc{$\mathtt{\{} \delta \mathtt{\}}\ \mathtt{while}\ \kappa\
    \mathtt{do}\ S_1; S_2; \cdots; S_m\ \mathtt{done}\
    \mathtt{\{} \epsilon \mathtt{\}}$ : an annotated loop}
  \KwOut{a loop invariant for the annotated loop}
  $\underline{\iota}$ := $\delta \vee \epsilon$\;
  $\overline{\iota}$ := $\epsilon \vee \kappa$\;
  \shade{$P := \mathtt{InitialAtomicPredicates()}$} \\
  \Repeat{a loop invariant is found}
  {
    \shade{\textbf{try}}

    \ \ \ \ \textbf{call} a learning algorithm for Boolean formulae
    where membership and\\
    \ \ \ \ equivalence queries are resolved by
    Algorithms~\ref{algorithm:membership} 
    and~\shade{\ref{algorithm:equivalence-query-resolution}} respectively\;
    \shade{\textbf{catch}
      $\mathtt{ConflictAbstractCEX} \rightarrow$} \\
    \shade{\ \ \ \ find distinct valuations $\nu$ and $\nu'$ in $\CEX$
      such that $\alpha^* (\nu) = \alpha^* (\nu')$\;}
    \shade{\ \ \ \ $P$ := $P \cup \mathtt{PredicatesFromConflict}
      (\nu, \nu')$\;} \\
    \shade{\textbf{catch}
      $\mathtt{ExcessiveRandomAnswers}(\theta) \rightarrow$} \\
    \shade{\ \ \ \ $P$ := $P \cup \mathtt{PredicatesFromConjecture}
      (\theta)$;} \\
    \shade{$\tau$ := $\lceil 1.3^{|P|} \rceil$;} \\
  }
  \phantom{x}
  \caption{Main Loop with Predicate Generation}
  \label{algorithm:main-loop-new}
\end{algorithm}

We first compute the initial set of atomic predicates by interpolating
$\underline{\iota}$ and $\neg \overline{\iota}$
(Section~\ref{section:initial-predicates}). With the initial set, we
start the learning process until the algorithm finds a loop invariant
or there is an exception raised. Exceptions basically mean that the
current set of predicates might not be enough to find a loop
invariant. We need in this case to find more predicates using one of
the algorithms explained in
Section~\ref{section:predicate-generation-by-interpolation}. 

The learning algorithm finds conflicting abstract counterexamples when
the equivalence query resolution algorithm gives a random
counterexample that contradicts the previous ones or the current
predicate abstraction is too coarse. Since we cannot distinguish the
two, we always generate more predicates using
Algorithm~\ref{algorithm:generating-predicate-conflicting-counterexamples},
hoping that we can find a loop invariant in the next iteration.

The $\mathtt{ExcessiveRandomAnswers}$ exception is raised when our new
equivalence query resolution algorithm, which is detailed later,
suspects that it generates too many random counterexamples because of
the coarse predicate abstraction. In this case, we generate more
predicates using Algorithm~\ref{algorithm:generating-predicates}.

Note that we start the learning algorithm from the scratch every time
we generate more predicates. The reason is because we use CDNF
algorithm for learning that handles only a fixed number of Boolean
variables. Recently, Chen \etal~\cite{CW:12:LBFI} propose a variant
of CDNF algorithm that supports incremental learning. We can also
adopt this algorithm to improve the efficiency of the overall
technique.

\begin{algorithm}[h]
  \tcc{$\CEX$ : a set of counterexamples}
  \tcc{$\tau$ : a threshold to generate new atomic predicates}
  \tcc{$\mathtt{\{} \delta \mathtt{\}}\ \mathtt{while}\ \kappa\
    \mathtt{do}\ S_1; S_2; \cdots; S_m\ \mathtt{done}\
    \mathtt{\{} \epsilon \mathtt{\}}$ : an annotated loop}
  \tcc{$\underline{\iota}, \overline{\iota}$ : under- and
    over-approximations to loop invariants}
  \KwIn{an equivalence query $\EQ (\beta)$ with $\beta \in
    \Bool[\obp]$} 
  \KwOut{$\YES$ or an abstract counterexample}
  $\theta$ := $\gamma (\beta)$\;
  \lIf{$\delta \Rightarrow \theta$ and $\theta \Rightarrow \epsilon
    \vee \kappa$ and $\theta \wedge \kappa \Rightarrow \precond
    (\theta : S_1; S_2; \cdots; S_m)$}
  \Return{$\YES$}\;
  \uIf{$\nu \models \neg (\underline{\iota} \Rightarrow \theta)$ or
    $\nu \models \neg (\theta \Rightarrow \overline{\iota})$ or
    $\nu \models \neg (\theta \wedge \kappa \Rightarrow \precond
    (\overline{\iota} : S_1; S_2; \cdots; S_m))$}
  {\shade{$\CEX := \CEX \cup \set{\nu}$\;} \Return{$\alpha^* (\nu)$}\;}
  \shade{\textbf{if}~the number of random abstract counterexamples
    $\leq \tau$~\textbf{then}} \\
  \ \ \ \ \ \Return{a random abstract counterexample}\;
  \shade{\textbf{else}} \\
  \ \ \ \ \ \shade{\textbf{throw}~$\mathtt{ExcessiveRandomAnswers}(\theta)$\;}
  \medskip
  \caption{Equivalence Query Resolution with Predicate Generation}
  \label{algorithm:equivalence-query-resolution}
\end{algorithm}

The equivalence query resolution algorithm is given in
Algorithm~\ref{algorithm:equivalence-query-resolution}. Again, we put
gray-boxes to denote the modified parts. As
Algorithm~\ref{algorithm:equivalence}, the mechanical teacher first
checks if the concretization of the abstract conjecture is a loop
invariant. If so, it returns $\YES$ and concludes the loop invariant
inference algorithm. Otherwise, the mechanical teacher compares the
concretization of the abstract conjecture with approximations to loop
invariants. If the concretization is stronger than the
under-approximation, weaker than the over-approximation, or it does
not satisfy the necessary condition given in
Proposition~\ref{proposition:necessary-over-approximation}, an
abstract counterexample is returned after recording the witness
valuation~\cite{VMCAI10,APLAS10}. The witnessing valuations are needed
to synthesize atomic predicates in
Algorithm~\ref{algorithm:main-loop-new} when conflicts occur.

If the concretization is not a loop invariant and falls between both
approximations to loop invariants, there are two possibilities. The
current set of atomic predicates is sufficient to express a loop
invariant; the learning algorithm just needs a few more iterations to
infer a solution. Or, the current atomic predicates are insufficient
to express any loop invariant; the learning algorithm cannot derive a
solution with these predicates. Since we cannot tell which scenario
arises, a threshold is deployed heuristically. If the number of random
abstract counterexamples is less than the threshold, we give the
learning algorithm more time to find a loop invariant. Only when the
number of random abstract counterexamples exceeds the threshold, can
we synthesize more atomic predicates for abstraction
refinement. Intuitively, the current atomic predicates are likely to
be insufficient if lots of random abstract counterexamples have been
generated. In this case, we raise $\mathtt{ExcessiveRandomAnswers}$
exception to synthesize more atomic predicates from the incorrect
conjecture in Algorithm~\ref{algorithm:main-loop-new}. Observe that in
Algorithm~\ref{algorithm:main-loop-new}, threshold $\tau$ is set to
$\lceil 1.3^{|P|} \rceil$, the approximate size of the search space,
which we found empirically.

\section{Experimental Results}
\label{section:experimental-results}

We have implemented the proposed technique in OCaml. 
In our implementation, the SMT solver \textsc{Yices} and the
interpolating theorem prover \textsc{CSIsat}~\cite{csisat} are used
for query resolution and interpolation respectively.  In addition to
the examples in~\cite{VMCAI10}, we add two more examples:
\texttt{riva} is the largest loop expressible in our simple language
from Linux\footnote{In Linux 2.6.30
  \texttt{drivers/video/riva/riva\_hw.c:nv10CalcArbitration()}}, and
\texttt{tar} is extracted from Tar\footnote{In Tar 1.13
  \texttt{src/mangle.c:extract\_mangle()}}.  All examples are
translated into annotated loops manually.  Data are the average of 100
runs and collected on a 2.4GHz Intel Core2 Quad CPU with 8GB memory
running Linux 2.6.31 (Table~\ref{table:experiments}).

\begin{table}[t]
  \scriptsize
  \caption{Experimental Results.\newline
    \small
    $P$ : \# of atomic predicates,
    $\MEM$ : \# of membership queries,  
    $\EQ$ : \# of equivalence queries, 
    $\mathit{RE}$ : \# of the learning algorithm restarts,
    $T$ : total elapsed time (s).
  }
  \centering
  \begin{tabular}{|c|r||r|r|r|r|r||r|r|r|r|r||r|r|}
    \hline
    \multirow{2}*{case} & \multirow{2}*{$SIZE$} &
    \multicolumn{5}{|c|}{\textsc{Previous~\cite{VMCAI10}}} &
    \multicolumn{5}{|c|}{\textsc{Current}} &
    \multicolumn{2}{|c|}{\textsc{BLAST~\cite{McMillan06}}}\\
    \cline{3-14}
    & & $P$ & $\MEM$ & $\EQ$ & $\mathit{RE}$ &
    $T$ & $P$ & $\MEM$ & $\EQ$ & $\mathit{RE}$ &
    $T$ & $P$ & $T$ \\
    \hline 
    \texttt{ide-ide-tape}     & 16&  6& 13& 7& 1&  0.05&  4&  6&  5&
    1&  0.05 & 21 & 1.31(1.07)\\ 
    \hline
    \texttt{ide-wait-ireason} &  9&  5& 790& 445&  33&  1.51&  5&122&
    91& 7& 1.09 & 9 & 0.19(0.14)\\ 
    \hline
    \texttt{parser}           & 37& 17& 4,223& 616&  13& 13.45&  9&
    86& 32&  1& 0.46 & 8 & 0.74(0.49)\\ 
    \hline
    \texttt{riva}             & 82& 20& 59& 11& 2& 0.51& 7& 14& 5& 1&
    0.37 & 12 & 1.50(1.17)\\ 
    \hline
    \texttt{tar}              &  7&  6&$\infty$&$\infty$&$\infty$&$\infty$& 
    2 & 2& 5& 1& 0.02 & 10 & 0.20(0.17)\\ 
    \hline
    \texttt{usb-message}      & 18& 10& 21& 7& 1& 0.10& 3& 7&  6&
    1&  0.04 & 4 & 0.18(0.14)\\ 
    \hline
    \texttt{vpr}              &  8& 5& 16& 9& 2& 0.05&  1&  1& 3&
    1&  0.01 & 4 & 0.13(0.10)\\ 
    \hline
  \end{tabular}
\label{table:experiments}
\end{table}

In the table, the column \textsc{Previous} represents the work
in~\cite{VMCAI10} where atomic predicates are chosen
heuristically. Specifically, all atomic predicates in pre- and
post-conditions, loop guards, and conditions of \texttt{if} statements
are selected. The column \textsc{Current} gives the results for our
automatic predicate generation technique. Interestingly, heuristically
chosen atomic predicates suffice to infer loop invariants for
all examples except \texttt{tar}. For the \texttt{tar} example, the
learning-based loop invariant inference algorithm fails to find a loop
invariant due to ill-chosen atomic predicates. In contrast, our
new algorithm is able to infer a loop invariant for the \texttt{tar}
example in 0.02s.
The number of atomic predicates can be significantly reduced as well.
Thanks to a smaller number of atomic predicates, loop invariant inference
becomes more economical in these examples. Without predicate generation,
four of the six examples take
more than one second. Only one of these examples takes
more than one second using the new technique. Particularly, the
\texttt{parser} example is improved in orders of magnitude.

The column \textsc{BLAST} gives the results of lazy abstraction
technique with interpolants implemented in
\textsc{BLAST}~\cite{McMillan06}. In addition to the total elapsed
time, we also show the preprocessing time in parentheses. Since the
learning-based framework does not construct abstract models,
our new technique outperforms \textsc{BLAST} in all cases but one
(\texttt{ide-wait-ireason}). If we disregard the time for
preprocessing in \textsc{BLAST}, the learning-based technique still
wins three cases (\texttt{ide-ide-tape}, \texttt{tar}, \texttt{vpr}) and 
ties one (\texttt{usb-message}). Also note that the number of atomic
predicates generated by the new technique is always smaller except
\texttt{parser}. Given the simplicity of the learning-based
framework, our preliminary experimental results suggest a promising
outlook for further optimizations.

\subsection{\texttt{tar} from Tar}
\label{subsection:tar}

\begin{figure}[t]
  \begin{equation*}
    \begin{array}{l}
      \assertion{\mathit{size} = M \land \mathit{copy} = N} \\
      1\ \mathtt{while}\ \mathit{size} > 0\ \mathtt{do} \\
      2\ \tab\mathit{available} := \mathtt{nondet}; \\
      3\ \tab\mathtt{if}\ \mathit{available} > \mathit{size}\
      \mathtt{then} \\
      4\ \tab\tab\mathit{copy} := \mathit{copy} + \mathit{available}; \\
      5\ \tab\tab\mathit{size} := \mathit{size} - \mathit{available}; \\
      6\ \mathtt{done}\\
      \assertion{\mathit{size} = 0 \implies \mathit{copy} = M + N}
    \end{array}
  \end{equation*}
  \caption{A Sample Loop in Tar}
  \label{figure:tar}
\end{figure}

This simple fragment is excerpted from the code for copying two
buffers. $M$ items in the source buffer are copied to the target
buffer that already has $N$ items. The variable $\mathit{size}$ keeps
the number of remaining items in the source buffer and $\mathit{copy}$
denotes the number of items in the target buffer after the last
copy. In each iteration, an arbitrary number of items are copied and
the values of $\mathit{size}$ and $\mathit{copy}$ are updated
accordingly.

Observe that the atomic predicates in the program text cannot express
any loop invariant that proves the specification. However, our new
algorithm successfully finds the following loop invariant in this
example:

\begin{equation*}
  M + N \le \mathit{copy} + \mathit{size} \land \mathit{copy} +
  \mathit{size} \le M + N
\end{equation*}

The loop invariant asserts that the number of items in both
buffers is equal to $M + N$. It requires atomic
predicates unavailable from the program text. Predicate generation is
essential to find loop invariants for such tricky loops.

\subsection{\texttt{parser} from SPEC2000 Benchmarks}
\label{subsection:parser}

\begin{figure}[b]
  \begin{equation*}
    \begin{array}{l}
      \mathtt{\{}\ \mathit{phase} = \false \wedge \mathit{success} = \false
      \wedge \mathit{give\_up} = \false \wedge \mathit{cutoff} = 0
      \wedge \mathit{count} = 0 \ \mathtt{\}}\\
      \;\; 1\ \mathtt{while}\ \neg(\mathit{success} \vee \mathit{give\_up})\
      \mathtt{do}\\
      \;\; 2\ \ \ \ \ \mathit{entered\_phase}\ \mathtt{:=}\ \false\mathtt{;}\\
      \;\; 3\ \ \ \ \ \mathtt{if}\ \neg \mathit{phase}\ \mathtt{then}\\
      \;\; 4\ \ \ \ \ \ \ \ \ \mathtt{if}\ \mathit{cutoff} = 0\ \mathtt{then}\
      \mathit{cutoff}\ \mathtt{:=}\ 1\mathtt{;}\\
      \;\; 5\ \ \ \ \ \ \ \ \ \mathtt{else}\ \mathtt{if}\ \mathit{cutoff} = 1 \wedge
      \mathit{maxcost} > 1\ \mathtt{then}\ \mathit{cutoff}\
      \mathtt{:=}\ \mathit{maxcost}\mathtt{;}\\
      \;\; 6\ \ \ \ \ \ \ \ \ \ \ \ \ \ \ \ \mathtt{else}\ 
      \mathit{phase}\ \mathtt{:=}\ \true\mathtt{;}\ 
      \mathit{entered\_phase}\ \mathtt{:=}\ \true\mathtt{;}\ 
      \mathit{cutoff}\ \mathtt{:=}\ 1000\mathtt{;}\\
      \;\; 7\ \ \ \ \ \ \ \ \ \mathtt{if}\ \mathit{cutoff} = \mathit{maxcost}
      \wedge \neg \mathit{search}\ \mathtt{then}\ 
      \mathit{give\_up}\ \mathtt{:=}\ \true\mathtt{;}\\
      \;\; 8\ \ \ \ \ \mathtt{else}\\
      \;\; 9\ \ \ \ \ \ \ \ \ \mathit{count}\ \mathtt{:=}\ \mathit{count} + 1\mathtt{;}\\
      10\ \ \ \ \ \ \ \ \ \mathtt{if}\ \mathit{count} > \mathit{words}\
      \mathtt{then}\ \mathit{give\_up}\ \mathtt{:=}\ \true\mathtt{;}\\
      11\ \ \ \ \ \mathtt{if}\ \mathit{entered\_phase}\ \mathtt{then}\
      \mathit{count}\ \mathtt{:=}\ 1\mathtt{;}\\
      12\ \ \ \ \ \mathit{linkages}\ \mathtt{:=}\ \mathtt{nondet;}\\
      13\ \ \ \ \ \mathtt{if}\ \mathit{linkages} > 5000\ \mathtt{then}\ 
      \mathit{linkages}\ \mathtt{:=}\ 5000\mathtt{;}\\
      14\ \ \ \ \ \mathit{canonical}\ \mathtt{:=}\ 0\mathtt{;}\ 
      \mathit{valid}\ \mathtt{:=}\ 0\mathtt{;}\\
      15\ \ \ \ \ \mathtt{if}\ \mathit{linkages} \neq 0\ \mathtt{then}\\
      16\ \ \ \ \ \ \ \ \ \mathit{valid}\ \mathtt{:= nondet;}\\
      17\ \ \ \ \ \ \ \ \ \mathtt{assume}\ 0 \leq \mathit{valid} \wedge
      \mathit{valid} \leq \mathit{linkages}\mathtt{;}\\
      18\ \ \ \ \ \ \ \ \ \mathit{canonical}\ \mathtt{:=}\
      \mathit{linkages}\mathtt{;}\\
      19\ \ \ \ \ \mathtt{if}\ \mathit{valid} > 0\ \mathtt{then}\
      \mathit{success}\ \mathtt{:=}\ \true\mathtt{;}\\
      20\ \mathtt{done}\\
      \mathtt{\{}\ 
      (\mathit{valid} > 0 \vee \mathit{count} > \mathit{words} \vee
      (\mathit{cutoff} = \mathit{maxcost} \wedge \neg
      \mathit{search})) \wedge\\
      \ \ \mathit{valid} \leq \mathit{linkages} \wedge
      \mathit{canonical} = \mathit{linkages} \wedge
      \mathit{linkages} \leq 5000
      \ \mathtt{\}}
    \end{array}
  \end{equation*}
  \caption{A Sample Loop in SPEC2000 Benchmark PARSER}
  \label{figure:parser}
\end{figure}

For the \texttt{parser} example (Figure~\ref{figure:parser}), 9 atomic
predicates are generated. These atomic predicates 
are a subset of the 17 atomic predicates from the program text.
Every loop invariant
found by the loop invariant inference algorithm contains all 9 atomic
predicates. This suggests that there are no redundant predicates. 
Few atomic predicates make loop invariants easier to comprehend.
For instance, the following loop invariant summarizes the condition when
$\mathit{success}$ or $\mathit{give\_up}$ is true:
\begin{equation*}
  \begin{array}{l}
    (\mathit{success} \vee \mathit{give\_up}) \timplies \\
    \hspace{.25in}
    (\mathit{valid} \ne 0 \vee
    \mathit{cutoff} = \mathit{maxcost} \vee
    \mathit{words} < \mathit{count}) \wedge\\
    \hspace{.25in}
    (\neg \mathit{search} \vee
    \mathit{valid} \ne 0 \vee 
    \mathit{words} < \mathit{count}) \wedge\\
    \hspace{.25in}
    (\mathit{linkages} = \mathit{canonical} \wedge
    \mathit{linkages} \ge \mathit{valid} \wedge
    \mathit{linkages} \le 5000)
  \end{array}
\end{equation*}

The invariant is simpler and thus easier to understand than the one
presented in~\cite{VMCAI10}. The right side of the implication
summarizes the condition when $\mathit{success}$ or
$\mathit{give\_up}$ becomes true.

Fewer atomic predicates also lead to a smaller standard deviation of
the execution time. The execution time now ranges 
from 0.36s to 0.58s with the standard deviation equal to 0.06. In contrast,
the execution time for~\cite{VMCAI10} ranges from 1.20s to
80.20s with the standard deviation equal to 14.09. 
By Chebyshev's inequality, the new algorithm infers a loop invariant
in one second with probability greater than $0.988$.
With a compact set of atomic predicates, loop invariant inference
algorithm performs rather predictably.

\subsection{\texttt{ide-wait-ireason} from Linux Device Driver}
\label{subsection:ide-wait-ireason}
\begin{figure}[h]
  \begin{equation*}
    \begin{array}{l}
      \assertion{\mathit{retries} = 100 \land
        (\lnot \mathit{ireason\_has\_ATAPI\_COD} \lor
        \mathit{ireason\_has\_ATAPI\_IO})} \\
      1\ \mathtt{while}\ \mathit{retries} \ne 0 \land 
      (\lnot \mathit{ireason\_has\_ATAPI\_COD} \lor
      \mathit{ireason\_has\_ATAPI\_IO})\ 
      \mathtt{do} \\
      2\ \tab\mathit{retries} := \mathit{retries} - 1; \\
      3\ \tab\mathit{ireason\_has\_ATAPI\_COD} := \mathtt{nondet}; \\
      4\ \tab\mathit{ireason\_has\_ATAPI\_IO} := \mathtt{nondet}; \\
      5\ \tab\mathtt{if}\ \mathit{retries} = 0\ \mathtt{then} \\
      6\ \tab\tab\mathit{ireason\_has\_ATAPI\_COD} := \true; \\
      7\ \tab\tab\mathit{ireason\_has\_ATAPI\_IO} := \false; \\
      8\ \mathtt{done}\\
      \assertion{\mathit{retries} < 100 \land
        \mathit{ireason\_has\_ATAPI\_COD} \land
        \lnot \mathit{ireason\_has\_ATAPI\_IO}}
    \end{array}
  \end{equation*}
  \caption{A Sample Loop in Linux IDE Driver}
  \label{figure:ide-wait-ireason}
\end{figure}

In the \texttt{ide-wait-ireason} example
(Figure~\ref{figure:ide-wait-ireason}), predicate generation 
performs better even though it generates the same number of atomic
predicates. This is because the
technique can synthesize the atomic predicate $\mathit{retries} \le
100$ which does not appear in the program text but is essential to
loop invariants. Surely this atomic predicate is expressible by the two
atomic predicates $\mathit{retries} = 100$ and $\mathit{retries} <
100$ from the program text. However the search space is significantly
reduced with the more succinct atomic predicate $\mathit{retries} \leq
100$. Subsequently, the learning algorithm only needs a quarter of
queries to infer a loop invariant.


\section{Conclusions}
\label{section:conclusions}

A predicate generation technique for learning-based loop invariant
inference was presented. The technique applies the interpolation
theorem to synthesize atomic predicates implicitly implied by program
texts. To compare the efficiency of the new technique, 
examples excerpted from Linux, SPEC2000, and Tar source codes were
reported. The learning-based loop invariant inference algorithm
is more effective and performs much better in these realistic examples. 

More experiments are always needed. Especially, we would like to have
more realistic examples which require implicit predicates unavailable in
program texts. Additionally,
loops manipulating arrays often require quantified loop invariants
with linear inequalities. Extension to quantified loop invariants is
also important.


\bibliographystyle{splncs-sorted}
\bibliography{ap-gen-lmcs}

\end{document}